\documentclass[11pt]{amsart}
\usepackage{amssymb}
\newtheorem{theorem}{Theorem} 
\newtheorem{proposition}{Proposition}

\newtheorem{corollary}{Corollary}

\newcommand{\latset}{{\tt LoI}}
\newcommand{\join}{\sqcup}
\newcommand{\meet}{\sqcap}
\newcommand{\ssm}[1] {${ {#1} }$}

\newcommand{\por}{\sqsubseteq}
\newcommand{\sem}[1]{\mbox{$[\![ #1 ]\!]$}}

\title[Algebraic Foundations for Quantitative Information Flow]{Algebraic Foundations for Information Theoretical, Probabilistic and Guessability measures of Information Flow }
\author[Pasquale Malacaria]{Pasquale Malacaria\\
School of Electronic Engineering and Computer Science\\
Queen Mary University of London\\
London, Mile End Road, E1 4NS, UK }

\begin{document}
\maketitle

\begin{abstract}
Several mathematical ideas have been investigated for Quantitative Information Flow.
Information theory, probability, guessability are the main ideas in most proposals. They aim to quantify {\em how much information} is leaked, {\em how likely is to guess} the secret and {\em how long does it take} to guess the secret respectively.
In this paper, we show how  the Lattice of Information provides a valuable foundation for all these approaches; not only it provides an elegant algebraic framework for the ideas, but also to investigate their relationship. In particular we will use this lattice to prove some results establishing order relation correspondences between the different quantitative approaches. The implications of these results w.r.t. recent work in the community is also investigated.
While this work concentrates on the foundational importance of  the Lattice of Information its  practical relevance has been recently proven, notably with the quantitative analysis of Linux kernel vulnerabilities. Overall we believe these works set the case for establishing the Lattice of Information as one of the main reference structure for Quantitative Information Flow.

\end{abstract}

\section{Introduction}

Quantitative security analysis should be able to address confidentiality\footnote{In this work we restrict ourselves to security as confidentiality} comparison questions like: ``given programs $P$ and $P'$ which one is more of a threat?"
This comparison problems is related to the other fundamental question that a quantitative security analysis should be able to address: ``how much of a threat is program $P$?"

%In general we could expect that by answering numerically the second question for two programs $P$ and $P'$ we should be able to answer the comparison problem. On the other hand in many cases the security of a program $P$ could be assessed by comparing it to some threshold program $P'$.

Quantitative analyses are based on some measure, usually a real number.
 This number may answer the comparison problems by reducing it to a numerical comparison and the second question by considering the magnitude of the number in relation to the size of the secret.  In many of these measures the number 0 has been shown to characterise secure programs.

In recent years a number of ideas have emerged as reasonable measures for Quantitative Information Flow (abbreviated as QIF): Information Theory, probabilistic  measures and guessability  \cite{0dav, belief, koepfbasin, Fossacs09}.
The information theoretical concepts of entropy, conditional entropy and mutual information have been used to answer questions like ``how much information can an attacker gain from observing the system?" whereas probabilities can be used to answer questions like ``how likely is that the attacker may guess the secret in $n$ tries after observing the system?" and guessability measures the question "what is the number of guesses needed to guess the secret after the observations?"

There seems to be an intuitive  connection between these questions, but the connection is not trivial; in fact  some deep differences have been noticed in these approaches \cite{Fossacs09}. In the context of QIF
the differences seems mainly to relate to the variety of attackers models and of what the scope of modelling should be.

In this work we aim to relate the confidentiality comparison questions in probabilistic, guessability and information theoretical approaches. We will do this by studying their relation to an algebraic structure: 
the Lattice of Information (abbreviated as \latset).

The Lattice of Information is the lattice of all equivalence relations on a set; by identifying  observations over a system as the equivalence relation equating all (secret) states that cannot be distinguished by those observations we see \latset\ as the mathematical model for all observations  generated by all possible deterministic systems over a set of (secret) states.

This allows for an elegant analysis decomposition of QIF into two steps, the first being an {\em algebraic interpretation}, the second being a {\em numerical evaluation}:
\begin{enumerate}
\item interpret the attacker view of the system as an {\em equivalence relation} identifying the states indistinguishable by the attacker through the observations,
\item {\em measure} the above equivalence relation. This measure should provide an indication of the leakage of confidential information (or vulnerability) of the system.
\end{enumerate}
 While these equivalence relations have been successfully  used in recent years \cite{0dav,JCS09, koepfbasin}, we aim here to prove some fundamental results about their algebraic structure.

 \begin{quote}
Given two systems $S,S'$ and the associated equivalence relations $\simeq_S,\simeq_{S'}$ we will show the following equivalences:
\begin{enumerate}
\item $\simeq_{S'}$ refines $\simeq_{S}$
\item  the leakage of $S$ is always less than the leakage of $S'$ (leakage measured by Shannon entropy).
\item  the expected probability of guessing the secret in $n$ tries according to $\simeq_{S}$ is always less than the expected probability of guessing the secret in $n$ tries according to $\simeq_{S'}$
\item  the expected numbers of guesses needed to guess the secret according to $\simeq_{S'}$ is always less than the expected numbers of guesses needed to guess the secret according to $\simeq_{S}$

\end{enumerate}
\end{quote}

In other terms given  two programs $P,P'$  to determine whether $P'$ refines $P$ (as observational equivalence relations) is the same as to determine whether is always the case that it is more likely to guess the secret using $P'$ instead of $P$. This is also the same as to determine whether the entropy of $P$ is always less than the entropy of $P'$. Moreover these results are shown to be consistent with different definitions of Quantitative Information Flow based on the adversary gain through observations i.e. the difference in threat before and after observations are made \cite{safety}.

These results hence provide a clear connection between the algebraic, probabilistic and information theoretical view of leakage.

The work  also contributes to the foundations of Quantitative Information Flow, in particular to the important work by G. Smith \cite{Fossacs09}, where the difference between the "one guess" model and the information theoretical one were insightfully debated.
What Smith noticed was that there exist programs such that, assuming a uniform distribution of the secret, their
 information theoretical measure is the same but whose vulnerability to a one guess attack is very different. In the argument it is important to consider a specific (in this case uniform) distribution.
 It is arguable however that code analysis should be affected by an element independent of the code, in this case the distribution.
 What our result shows is that if we argue about the relative vulnerability to $n$ tries attack of two programs and the argument is not dependent on a specific distribution then their relative vulnerability is determined by their \latset\ order or equivalently by their entropy order.
  
The algebraic aspect of QIF, i.e. the \latset\ interpretation of programs is far from being a pure academic exercise; in fact it has informed works integrating QIF with verification techniques  \cite{koepf09, acsac} where model checkers and sat-solvers are used to build the equivalence $\simeq_S$ associated to a program.
More recently these ideas  have been exploited  to build the first quantitative analysis for 
 real code leakage, in particular to quantify leakage of Linux kernel functions \cite{acsac}.
These works make use of a basic relation between \latset\ and Information Theory, i.e. the fact that  $\log(|\simeq_S|)$ is the channel capacity of the system $S$, i.e. the  maximum amount that $S$ can leak.

\section{Basics}\label{BASIC}

\subsection{Observations and the lattice of information}
We can see {\em observations} over a system as some {\em partial } information on systems' states, in that an observation reveals some information about the states of the system.  
Some systems may allow for observations revealing no information (all states are possible according to that system's observations) while other systems may allow for observations revealing complete information on the states of the system.

We will make an important determinacy  assumption about observations, i.e. that a system's
observations  form a {\em partition} on the set of all possible states: a block in this partition is the set of states that are indistinguishable by that observation. 
This assumption is satisfied for example in the setting of sequential languages when we take as observations the program outputs
because the inverse image of a function form a partition on the function domain.

In this work we will use the terms partition or equivalence relation interchangeably. An equivalence relation can always be seen as the partition whose blocks are the equivalence classes and a partition  can always be seen as the equivalence relation defined by two objects are related iff they are in the same block.

\subsection{Partitions and equivalence relations as lattice points}

Given  a finite set  $\Sigma$
the set of all possible equivalence relations over  $\Sigma$ is a {\em complete lattice}: the Lattice of Information (abbreviated as \latset) \cite{loi}. Order on equivalence relations is the refinement order.

Formally let us define the set \latset\ as the set of all possible equivalence relations on a set $\Sigma$. 
Given  $\approx,\sim$  $\in $ \latset\ and $\sigma_{1},\sigma_{2} \in \Sigma$
the ordering of \latset\ is defined as
\begin{equation}\label{ORDER}
\approx\ \sqsubseteq\ \sim\ \leftrightarrow\ \forall \sigma_1, \sigma_2\ (\sigma_1 \sim \sigma_2 \Rightarrow \sigma_1 \approx \sigma_2)
\end{equation}
This is a {\em refinement} order: classes in $\sim$ refine (split) classes in $\approx$.
Thus, higher elements in the lattice can distinguish more while lower elements in the lattice can distinguish less states.
It easily follows from  (\ref{ORDER}) that \latset\ is a complete lattice. 

Alternatively  the lattice operations join $\join$ and meet $\meet$  are defined as the intersection of relations and the transitive closure union of relations respectively. 

The restriction to consider finite lattices is motivated by considering information storable in programs variables: such information is $\le 2^k$ where $k$ is the number of bits of the secret variable.

In terms of partitions,  a partition is above another if it is more informative, i.e. each block in the lower partition is included in a block in the above partition

Here is an example of how these equivalence relations can be used in an information flow setting. Let us assume the set of states $\Sigma$ consists of a tuple $\langle l, h\rangle$ where $l$ is an observable, usually called \textit{low}, variable and $h$ is a confidential variable, usually called \textit{high}. One possible observer can be described by the equivalence relation
 \[ \langle l_1,h_1 \rangle \approx \langle l_2,h_2 \rangle \leftrightarrow l_1 = l_2 \]
That is the observer can only distinguish two states whenever they agree on the low variable part. Clearly, a more powerful attacker is the one who can distinguish any two states from one another, or
\[ \langle l_1,h_1 \rangle \sim \langle l_2,h_2 \rangle \leftrightarrow l_1 = l_2 \land h_1 = h_2 \]
The $\sim$-observer gains more information than the $\approx$-observer by comparing states, therefore $\approx\ \sqsubseteq\ \sim$.

\subsection{Lattice of information as a lattice of random variables}\label{SECRV}

A random variable (noted r.v.) is usually defined as a map $X : D \rightarrow \mathbb R$, where $D$ is a finite set with a probability distribution and the real numbers $\mathbb R$ is the range of $X$. For each element $d \in D$, its probability will be denoted $\mu(d)$. For every element $x \in \mathbb R$ we write $\mu(X=x)$ (or often in short $\mu(x)$) to mean the probability that $X$ takes on the value $x$, i.e. $\mu(x) \stackrel{def}{=} \sum_{d \in X^{-1}(x)}{\mu(d)}$. In other words, what we observe by $X = x$ is that the input to $X$ in $D$ belongs to the set $X^{-1}(x)$. From that perspective, $X$ partitions the space $D$ into sets which are indistinguishable to an observer who sees the value that $X$ takes on\footnote{We define an {\em event} for the random variable a block in the partition.}. This can be stated relationally by taking the kernel of $X$ which defines the following equivalence relation $\mbox{ker}(X)$:
\begin{equation}\label{KERNEL}
d\ \mbox{ker}(X)\ d'\ \mbox{iff } X(d) = X(d') 
\end{equation}

Equivalently we write $X \simeq Y$ whenever the following holds
%\begin{equation*}
\[ X \simeq Y \mbox{ iff } \{X^{-1}(x) : x \in \mathbb R\} = \{Y^{-1}(y) : y \in \mathbb R\} \]
%\end{equation*}
and thus if $X \simeq Y$ then $H(X) = H(Y)$.

This shows that each element of the lattice \latset\ can be seen as a random variable.

Given two r.v. $X,Y$ in \latset\  we define the joint random variable $(X,Y)$ as their least upper bound in \latset\ i.e. $X \join Y$. It is easy to verify that $X \join Y$ is the partition obtained by all possible intersections of blocks of $X$ with blocks of $Y$.

\subsection{Basic concepts of Information Theory}
This section contains a very short review of some basic definitions of Information Theory; additional background is readily available both in textbooks (the standard being Cover and Thomas textbook \cite{CoTho}).
Given a space of events with probabilities $P=(p_i)_{i\in N}$ ($N$ is a set of indices) the Shannon's entropy is defined as 

\begin{equation}
H(X) = - \sum_{i \in N}{p_i}\log{p_i}
\end{equation}
%\ssm{H(P)=-\Sigma_{i\in N} p_i {\log}(p_i)}. 
It is usually said that this number measures the average information content of the set of events: if there is an event with probability 1 then the entropy will be 0 and if the distribution is uniform i.e. no event is more likely than any other the entropy is maximal, i.e.  $\log \vert N \vert$. 
In the literature the terms information content and uncertainty in this context are often used interchangeably: both terms refer to the number of possible distinctions on the set of events in the sense we discussed before.

The entropy of a r.v. $X$ is just the entropy of its probability distribution i.e. 
\[{ - \sum_{x \in X}\mu(X=x) \log \mu(X=x)}\]
Given two random variables $X$ and $Y$, the joint entropy $H(X,Y)$ measures the uncertainty of the joint r.v. $(X,Y)$. it Is defined as
%\begin{equation*} 
\[ - \sum_{x \in X, y \in Y}\mu(X=x,Y=y) \log \mu(X=x,Y=y) \]
%\end{equation*}

Conditional entropy \ssm{H(X\vert Y)} measures the uncertainty about $X$ given knowledge of $Y$. 
It is defined as $H(X,Y)-H(Y)$.
The higher \ssm{H(X\vert Y)} is, the lower is the correlation between \ssm X and \ssm Y. It is easy to see that if \ssm X is a function of \ssm Y, then \ssm{H(X\vert Y)=0}  (there is no uncertainty on $X$ knowing $Y$ if $X$ is a function of $Y$) and if  \ssm X and \ssm Y are independent then \ssm{H(X\vert Y)=H(X)} (knowledge of $Y$ doesn't change the uncertainty on $X$ if they are independent) .

Mutual information $I(X;Y)$ is a measure of how much information $X$ and $Y$ share. It can be defined as
\[{ I(X;Y)=H(X)-H(X\vert Y)=H(Y)-H(Y\vert X)}\]
Thus the information shared between $X$ and $Y$ is the information of $X$ (resp $Y$) from which the information about $X$ given $Y$ has been deduced.
This quantity measures the correlation between \ssm X and \ssm Y. 
For example  \ssm X and \ssm Y are independent iff \ssm{I(X;Y)=0}.

Mutual information is a measure of binary {\em interaction}. 
Conditional mutual information, a form of ternary interaction will be used to quantify  {\em leakage}. Conditional mutual information measures the correlation between two random variables conditioned on a third random variable; it is defined as: 
%\begin{equation*}
\[ I(X;Y\vert Z)=H(X\vert Z)-H(X\vert Y,Z)=H(Y\vert Z)-H(Y\vert X,Z) \]
%\end{equation*}

\subsection{Measures on the lattice of information}

Suppose we want attempt to quantify the amount of information provided by a point in the lattice of information.

We could for example associate to a partition $P$  the measure $|P|=$
``number of blocks in $P$". This measure would be 1 for the least informative partition, its maximal value would be the number of atoms and would be reached by the top partition. 
It is also true that $A\por B$ implies $|A|\leq|B|$ so the measure reflects the order of the lattice. An important property of ``additivity" for measures is the inclusion-exclusion principle:  this principle says that things should not be counted twice. In terms of sets, the inclusion-exclusion principle says that the number of elements in a union of sets is the sum of the number of elements of the two sets minus the number of elements in the intersection.
The inclusion-exclusion  principle is universal e.g. in propositional logic the truth value of $A\vee B$ is given by the truth value of $A$ plus the truth value of $B$ minus the truth value of $A\wedge B$.

 in the case of the number of blocks the inclusion-exclusion principle is:
\[ |A\join B| =|A|+|B|-|A\meet B|\] 
Unfortunately this property does not hold. As example, by taking
\[ A=\{\{1,2  \}\{3,4\}\} ,\ B=\{\{1,3\}\{2,4\}\}  \]
as two partitions, then their join and meet will be
\[ A\join B=\{\{1\}\{2  \}\{3\}\{4\}\} ,\ A\meet B=\{\{1,3,2,4\}\}.  \]
hence $ |A\join B| =4 \not= 3=|A|+|B|-|A\meet B|$.

Another problem with the map $|\ |$ is that  when we consider \latset\ as a lattice of random variables the above measure may end up being too crude; in fact, all probabilities are disregarded\footnote{We will see however in later sections how the number of blocks relates to Information Theory and channel capacity } by  $|\ |$.
To address these problems more abstract
lattice theoretic notions have been introduced in the literature \cite{birkhoff}.

A valuation on \latset\ is a real valued map $\nu : $ \latset $\rightarrow \mathbb R$,  
that satisfies the following properties:
\begin{eqnarray}
\nu(X \join Y) \ = \nu(X) + \nu(Y)-\nu(X \meet Y) \label{REALINCLEXCL}\\
X \por Y\ \mbox{ implies }\ \nu(X) \le \nu(Y)\label{ORDERPRES}
\end{eqnarray}
A join semivaluation is a weak valuation, i.e. a real valued map satisfying 
\begin{eqnarray}
 \nu(X \join Y) \le \nu(X) + \nu(Y)-\nu(X \meet Y) \label{INCLEXCL}\\
X \por Y\ \mbox{ implies }\ \nu(X) \le \nu(Y)%\label{ORDERPRES}
\end{eqnarray}
for every element $X$ and $Y$ in a lattice \cite{birkhoff}.
%These are incidentally the two properties of our map $\nu_\join$ described earlier. 
The property (\ref{ORDERPRES}) is order-preserving: a higher element in the lattice has a larger valuation than elements below itself. The first property (\ref{INCLEXCL}) is a weakened inclusion-exclusion principle. 

\begin{proposition}\label{PROP_SEMIVAL}
Entropy is  join semivaluation on \latset\ by defining
\begin{equation}\label{semival1}
\nu (X\join Y)=H(X,Y) 
\end{equation}
\end{proposition}
\begin{proof}
Property \ref{ORDERPRES} is well known; for inequality \ref{INCLEXCL} start from the known equality 
\[H(X,Y)=H(X) + H(Y)-I(X;Y)   \]
it will be hence enough to prove that
\[ H(X \meet Y)\leq I(X;Y)   \]
This can be proved by noticing that 
\begin{enumerate}
\item $H(X\meet Y)=I(X\meet Y; X)$ this is clear because $I(X\meet Y; X)$ measure the information shared between $X\meet Y$ and $X$ and because $X\meet Y\por X$ such measure has to be $H(X\meet Y)$
\item $I(X\meet Y; X)\leq I(Y; X)$ this is clear because $X\meet Y\por Y$ hence there is more information shareable between $Y$ and $X$ than between $X\meet Y$ and $X$
 \end{enumerate}

 combining we have
 \[ H(X\meet Y)=I(X\meet Y; X)\leq I(Y; X) \]
 \end{proof}

\subsection{Note: Entropy as the best measure on \latset}
An important result  proved by Nakamura \cite{SemVal} gives a particular importance to Shannon entropy as a measure on \latset. He proved that the \textit{only} probability-based join semivaluation on the lattice of information is Shannon's entropy. It is easy to show that a valuation itself is not definable on this lattice, thus Shannon's entropy is the best approximation to a probability-based valuation on this lattice.

 Nakamura starts by  considering
a family of function $(f_n)_{n\in N}$ such that $f_n$ is defined on a set of $n$ probabilities $p_1,\dots,p_n$ and satisfies:
\begin{enumerate}
\item $f_n$ is continuous 
\item $f_n$ is permutation invariant, i.e. $f_n(p_1,\dots,p_n)=f_n(p_{\pi(1)},\dots,p_{\pi(n)})$ for any permutation $\pi$
\item $f_{n+1}(p_1,\dots,p_n,0)=f_n(p_1,\dots,p_n)$
\end{enumerate}
Such  a family $(f_n)_{n\in N}$ induces a function $F$ on partitions with $n$ blocks $X=\{X_1,\dots,X_n\}$ with block $X_i$ having probability $p_i$:
\[ F(X)=  f_n(p_1,\dots,p_n) \]
Suppose now that 
\begin{enumerate}
\item $F$ is a join-semivaluation on all lattices of partitions
\item If two partitions $X,Y$ are independent (in probability theory sense) then 
\[ F(X\join Y)=F(X)+F(Y)\]
\end{enumerate}

Nakamura's result is then that such a function $F$ is, up to a constant, Shannon's entropy function, i.e.
\[ F(X)=f_n(p_1,\dots,p_n)=-c\sum_{1\leq i\leq n} p_i \log(p_i)   \]

\section{Lattice of Information, expected probability of guessing, expected number of guesses and Entropy}\label{LAB:MAIN}

This section contains the main results of this article, i.e.  correspondence between the order relation of \latset, expected probability of guessing,  expected number of guesses and entropy.

\subsection{Expected probability of guessing}

We want to define, given an equivalence relation,  the average probability of guessing the secret in $n$ tries.

 Given a set $X$  where each element has associated a probability  (w.l.g. we assume the probabilities  being ordered decreasingly i.e. $ \mu(x_i)\ge \mu(x_{i+1})$)
define the probability of guessing the secret in $n$ tries as
\[  g_{n,\mu}(X)=\sum_{1\leq i\leq n} \mu(x_i) \]
Given a partition $X$ and a distribution $\mu$ the probability of guessing the secret in $n$ tries is
\[ G_{n,\mu}(X)=\sum_{X_i\in X} g_{n,\mu}(X_i) \]

As an example consider the partition 
\[ \{ \{ x_1,\dots,x_4\}\{ x_5,x_6\}  \} \]
where the first four atoms have probability $\frac{1}{16}$ each and $x_5,x_6$ have probability $\frac{3}{8}$ each.

Then the average probability of guessing the secret in $2$ tries is $\frac{1}{8}+\frac{3}{4}=\frac{7}{8}$; indeed after the observations and two tries the probability of non guessing the secret is $\frac{1}{8}$ corresponding to not having exhausted all possibilities from the first block.

Notice that the above definition is the same as having a probability distribution on each block, computing the probability of guessing the secret in each block and then taking the weighted average:

\[  G_{n,\mu}(X)=\sum_{X_i\in X} g_{n,\mu}(X_i) = \sum_{X_i\in X} \mu(X_i) \sum_{1\leq j\leq n, x_j\in X_i} \frac{\mu(x_j) }{\mu(X_i)}   \]

When clear from the context we will omit the subscript $\mu$ from $G$ and $g$.

\begin{theorem}\label{Lab:LoIExpProb}
 \[ X\sqsubseteq Y \Leftrightarrow \forall \mu, n. \ G_{n,\mu}(X)\leq G_{n,\mu}(Y) \]
 \end{theorem}
 \begin{proof}
 Step 1:
 \[ X\sqsubseteq Y \Rightarrow \forall \mu, n. \ G_{n,\mu}(X)\leq G_{n,\mu}(Y) \]
w.l.g. it will be enough to consider a block $X_i$ in $X$ splitting into two blocks $Y_i,Y_j$ in $Y$; we then need to prove that
\[ g_n( {X_i}) \leq g_n( {Y_i}) + g_n( {Y_j})  \]
We can write $g_n( {X_i})= \sum_{i\leq I}\mu( x_i)+\sum_{j\leq J} \mu(x_j)$ where the $x_i$ are elements in the block $Y_i$ and the $x_j$ are in $Y_j$ . We can then write $g_n( {Y_i})$ as $\sum_{i\leq I}\mu( x_i)+c_i$ where $c_i\ge 0$ is the sum of the elements in $g_n( {Y_i})$ which are not in $g_n( {X_i})$ and similarly   $g_n( {Y_j})$ can be written as $\sum_{j\leq J}\mu( x_j)+c_j$.

We have hence
\begin{eqnarray*}
g_n( {Y_i}) + g_n( {Y_j})&=&\\
\sum_{i\leq I} \mu( x_i)+c_i +\sum_{j\leq J} \mu( x_j)+c_j&\ge& \\
\sum_{i\leq I} \mu( x_i)+\sum_{j\leq J} \mu( x_j)&=& \\
 g_n( {X_i})&&
\end{eqnarray*}

Step 2: 
 \[ (\forall \mu, n.\ \  G_n(X)\leq G_n(Y))  \Rightarrow X\sqsubseteq Y\]
Reason by contradiction: suppose $X\not\sqsubseteq Y$, w.l.g. we can then find a block  $Y_i\in Y$ included in two (or more) blocks in $X$; We then take a distribution 0 everywhere apart from the elements in $Y_i$ and apply the previous reasoning, then for this distribution
$G_n(X)\not\leq G_n(Y)$ by taking  $n=|Y_i|-1$
\end{proof}

As an example consider the partitions 
\[ X=\{\{1,2\}\{3,4\}  \},\ \ \ Y=\{ \{1,3\}\{2,4\} \}		\]

$X$ and $Y$ are not order related because no block in $X$ is refined by a block in $Y$ and vice-versa; hence following the theorem we can find distributions and number of guesses ordering them in any order:
for $G(Y)<G(X)$  take the distribution giving $\frac{1}{2}$ to $1,3$ and 0 elsewhere; then $n=|\{1,3\}|-1=1$ and so  we have 
\[G_1(Y)=g_1( \{1,3\})=\frac{1}{2}<\frac{1}{2}+\frac{1}{2}=g_1(\{1,2\})+g_1(\{3,4\})=G_1(X) \]

Likewise for $G(X)<G(Y)$ choose the distribution giving $\frac{1}{2}$ to $1,2$ and 0 elsewhere.

\begin{corollary}\label{Lab:LoIExpProb2}
 \[ X\sqsubseteq Y \Leftrightarrow \forall \mu \ G_{1,\mu}(X)\leq G_{1,\mu}(Y) \]
 \end{corollary}
 \begin{proof}
  Direction $\Rightarrow$ is the same as theorem \ref{Lab:LoIExpProb};  direction $\Leftarrow$ is also similar: just notice that 
 the choice of $n=|Y_i|-1$ in $G_{n,\mu}$ implies $n\ge 1$ because as $Y_i$ is split in several blocks it must have at least 2 elements, hence we can replace $|Y_i|-1$ with 1
 \end{proof}

\subsection{Expected number of guesses}

The expected probability of guessing should be related to the expected number of guesses.

 Given a set $X$  where each element has associated a probability  (w.l.g. we assume the probabilities  being ordered decreasingly i.e. $ \mu(x_i)\ge \mu(x_{i+1})$)
define the expected number of guesses as
\[  NG_{\mu}(X)=\sum_{1\leq i\leq n} i \mu(x_i) \]
Given a partition $X$ and a distribution $\mu$ the expected number of guesses is (we abuse the notation):
\[ NG_{\mu}(X)=\sum_{X_i\in X} NG_{\mu}(X_i) \]

Intuitively the more is known of the secret the less guesses are needed, hence we should expect the $NG$ order to reverse the \latset\    order; consider for example the set $\{a,b,c,d\}$  with probabilities $ \frac{1}{2},\frac{1}{4},\frac{1}{8},\frac{1}{8} $ respectively; we have then
\[NG( \{ \{a,b,c,d\}  \} ) = \frac{15}{8} > \frac{10}{8} = NG( \{ \{a,d\}\{b,c\}  \}  )  \]

We can now show that \latset\ order is the dual of the $NG$ order:

\begin{theorem}
 \[ X\sqsubseteq Y \Leftrightarrow \forall \mu,\ NG_{\mu}(Y)\leq NG_{\mu}(X) \]
 \end{theorem}
\begin{proof}
\[ X\sqsubseteq Y \Rightarrow \forall \mu,\ NG_{\mu}(Y)\leq NG_{\mu}(X) \]

w.l.g. it will be enough to consider a block $X_i$ in $X$ splitting into two blocks $Y_i,Y_j$ in $Y$; 
consider an element $x\in X_i$; this element will appear as a term $j \mu(x)$ in the sum $NG_{\mu}(X_i)$.
As the elements of $X_i$ are split in the two sets $Y_i,Y_j$ then the same $x$ will appear in $NG(Y_i)$ or in $NG(Y_j)$: in any case it will appear as a term $j' \mu(x)$ where $j'\leq j$ because $X_i$ is split in the two sets $Y_i,Y_j$ so the relative order of $x$ in $Y_i$ or $Y_j$ has to be less than the relative order of $x$ in $X_i$. Hence the statement is true.

\[ X\sqsubseteq Y \Leftarrow \forall \mu,\ NG_{\mu}(Y)\leq NG_{\mu}(X) \]

Reason by contradiction: suppose $X\not\sqsubseteq Y$, w.l.g. we can then find a block  $Y_i\in Y$ included in two (or more) blocks in $X$; We then take a distribution 0  everywhere apart from the elements in $Y_i$ and apply the above reasoning: then for this distribution
$NG(Y)\not\leq NG(X)$
\end{proof}

\subsection{Entropy and \latset}
The next fundamental result is about entropy; again we can relate entropy to order in \latset. Two partitions are order related if and only if they are entropy related (in the same direction) for all possible distributions

\begin{theorem}\label{EntLoi}
 \[ X\sqsubseteq Y \Leftrightarrow \forall \mu,\ H_{\mu}(X)\leq H_{\mu}(Y) \]
 \end{theorem}
\begin{proof}
Step 1:
 \[ X\sqsubseteq Y \Rightarrow \forall \mu,\ H_{\mu}(X)\leq H_{\mu}(Y) \]
 This is a well known property of entropy: taking larger probabilities reduce entropy and it is a  consequence of the Jensen inequality
 
 Step 2:
  \[ X\sqsubseteq Y \Leftarrow \forall \mu,\ H_{\mu}(X)\leq H_{\mu}(Y) \]
Reason by contraposition; 

 suppose $X\not\sqsubseteq Y$, w.l.g. we can then find a block  $Y_i\in Y$ included in two (or more) blocks in $X$ (say $X_1\dots X_n$); We then take a distribution 0   everywhere apart from the elements in $Y_i$; notice that for such a distribution $\mu(Y_i)=1$ whereas in $X$ there are more two or more blocks with non zero probability: we have hence
 
 \[  H(X)=-\sum_{1\leq i\leq n} \mu(X_i)\log( \mu(X_i)) >  0=- \mu(Y_i)\log( \mu(Y_i))=H(Y)  \]
 \end{proof}

 \subsection{Shannon's order of information}

The Lattice of Information was pioneered in a little known note by Shannon \cite{Sh2} in order to characterise information. 

One of Shannon's motivations was that while Information Theory  is a measure of information it is not a characterisation of it.
Information Theory aims to measure the amount of information of random variables or of some sort of stochastic process: what the information is about is not a concern of the theory, the measure is based on the {\em number of distinctions} available in an information context. 
As an example consider the information-wise very different processes ``flipping a coin" and ``presidential election between two candidate". While the first is a rather inconsequential process and the second may have important consequences they are both contexts allowing for two choices hence they both have an information measure of  (at most) 1 bit. In a context where $n$ choices are possible (a process with $n$ outcomes) the information associated is measured in terms of the number of bits needed to encode those possible choices, so it is at most $\log_2(n)$. Hence completely different information contexts may result in the same information theoretical measure.

We can however try to characterise ``information" using Information Theory. In the above example while ``flipping a coin" and ``presidential election between two candidate" may have the same measure, it is not the case that knowing one of the two gives information about the other, so $H(X|Y)>0$ for $X,Y$ being one of  ``flipping a coin" or  ``presidential election".

Given random variables $X,Y$ Shannon's order is defined by:

\[ X\leq_d Y  \Leftrightarrow H(X|Y)=0 \]
 
The intuition here is that $Y$ provides complete information about $X$, or equivalently $X$ has less information than $ Y$, so $X$ is an abstraction of $Y$ (some information is forgotten).

Shannon also defined the related distance function:
\[ d(X,Y)=H(X|Y)+H(X|Y)\]

The function $d$ and the relation $\leq_d $ are related as follows:

\[ d(X,Y)=0 \Leftrightarrow  X\leq_d Y \wedge Y\leq_d X \]

In fact suppose $d(X,Y)=0$; then  $H(X|Y)+H(X|Y)=0$ so as conditional entropy is non negative $X\leq_d Y \wedge Y\leq_d X$.
On the other hand  $X\leq_d Y \wedge Y\leq_d X$ implies $H(X|Y)=0, H(Y|X)=0$ so $d(X,Y)=0$.

The equivalence classes of the order $\leq_d $, i.e. points s.t. $X\leq_d Y \wedge Y\leq_d X$ or equivalently  the sets of points of distance 0, are the information theoretical characterization of information: all items in a class can be seen as objects having {\em the same} information, not just sharing the same measure.

Shannon's order and \latset\ order are the same:

\begin{theorem}
\[ X \sqsubseteq Y\Leftrightarrow \forall \mu.\  X\leq_d Y \]
\end{theorem} 
\begin{proof}
Direction $ X \sqsubseteq Y\Rightarrow  \forall \mu. X\leq_d Y $:

By definition of join in a lattice 
\[ X \sqsubseteq Y\Leftrightarrow X\sqcup Y=Y \]
hence we have
\[ X \sqsubseteq Y\Leftrightarrow H(X,Y)=H( X\sqcup Y)=H(Y)\]
and so
\[  H(X|Y) = H(X,Y)-H(Y) =H(Y)-H(Y)=0 \]
which proves $\forall \mu.\  X\leq_d Y$

For the other direction assuming $X \not\sqsubseteq Y$ then $X\sqsubset X\sqcup Y $ so we can find a distribution s.t. 
$H( X\sqcup Y )>H(X)$ and so 
\[ H(Y|X)=H( X\sqcup Y )-H(X)> H(X)-H(X)=0  \]

and we conclude $X\not\leq_d Y$
\end{proof}

Shannon also noticed that $d$ defines a pseudometric and so the quotient space by the equivalence classes of points of distance 0 is a metric space.

%%%
\section{Measuring leakage of programs}\label{LEAKDEF}

Now we want to connect \latset\ with leakage of confidential information in programs.

\subsection{Observations over programs}
Observations over a program $P$ form an equivalence relation on states of $P$. A particular equivalence class will be called an observable. Hence an observable is a set of states indistinguishable by an attacker making that observation.

\newcommand{\Loi}{{\tt LoI }}
\newcommand{\s}[1]{{\tt {$#1 $}}}

The above intuition can be formalized in terms of several program semantics.
We will concentrate here on a specific class of observations: the output observations  \cite{JCS09,POPL07}.
For this observation the random variable associated to a program $P$ is the equivalence relation on any two  states $\sigma, \sigma'$ from the universe of states $\Sigma$ defined by
\begin{equation}\label{LoI(P)}
\sigma \simeq \sigma' \iff \sem{P}(\sigma) = \sem{P}(\sigma') 
\end{equation}
where $\sem{P}$ represents the denotational semantics of $P$. Hence the equivalence relation amounts to`` have the same observable output''. We denote the interpretation of a program $P$ in \latset\   as defined by the equivalence relation (\ref{LoI(P)}) by $\Loi(P)$. 
According to denotational semantics commands are considered as state transformers, informally maps which change the values of variables in the memory; similarly, language expressions are interpreted as maps from the memory to values. The equivalence relation $\Loi(P)$ is hence nothing else than the set-theoretical kernel of the denotational semantic of $P$. Assuming that the set of confidential inputs $h$ is equipped with a probability distribution $\mu$ we can see $\Loi_{\mu}(P)$ as a random variable. We will write simply $\Loi(P)$ unless we need to specify a specific distribution $\mu$.

\subsection{ \latset\ interpretation of programs and basic properties}

In this paper we will consider the well known {\bf while} programming language \cite{W93}, that is  a simple imperative language with assignments, sequencing,
conditionals and loops.  Syntax and semantics for the language are standard, as in e.g. \cite{W93}. The expressions of the language are arithmetic expression, with constants \ssm{0,1,\dots} and boolean expressions with constants \s{{\tt tt,ff}}.

To see a concrete example, let $P$ be the program
\begin{quote}
\begin{verbatim}
if (h==0) then x=0; else x=1;
\end{verbatim}
\end{quote}
 where the variable \texttt{h} ranges over $\{0,1,2,3\}$. 
 We will assume for the time being that in all program we consider the low variables are initialized in the code; we will discuss this assumption in section \ref{LOW}.

The equivalence relation (i.e. partition) $\Loi(P)$ associated to the above program is then
\begin{eqnarray*}
\Loi(P) &=& \{\underbrace{\{0\}}_{\mbox{{\tt x=0}}} \underbrace{\{1,2,3\}}_{{\tt x=1}} \}
\end{eqnarray*}
$\Loi(P)$ effectively partitions the domain of the variable $h$, where each disjoint subset represents an output. The partition reflects the idea of what an attacker can learn of secret inputs by \textit{backwards} analysis of the program, from the outputs to the inputs.

The quantitative evaluation of the partition $\Loi(P)$ measures such knowledge gains of an attacker, solely depending on the partition of states and the probability distribution of the input.

\subsection{Definition of leakage}
%\bigskip
Let us start from the following intuition
\begin{quote}
The leakage of confidential information of a program is defined as
the difference between an attacker's uncertainty about the secret before and after available observations about the program.
\end{quote}
 For a Shannon-based measure, the above intuition can be expressed in terms of conditional mutual information.
 In fact if we start by observing that the  attacker uncertainty about the secret before observations is $H(h|l)$
 and the  attacker uncertainty about the secret after observations is $H(h|l,\Loi(P))$ then using the definition of conditional mutual information we define leakage as
 \[ H(h|l) - H(h|l,\Loi(P)) = I(h;\Loi(P)|l)  \]
 
We can now simplify the above definition as follows
\begin{eqnarray}
I(\Loi(P);h\vert l) &=&  H(h|l) - H(h|l,\Loi(P)) \nonumber\\
&=_{\ }& H(\Loi(P)\vert l ) - H(\Loi(P)\vert l,h) \nonumber \\
&=_A& \ H(\Loi(P)\vert l ) - 0 \nonumber\\
&=& H(\Loi(P)\vert l ) \nonumber \\
&=_B&\ H(\Loi(P))\label{DEFLEAK}
\end{eqnarray}
where in the first equality we used the symmetry of conditional mutual information; the equality $A$ holds because the program is deterministic and $B$ holds when the program only depends on the high inputs, for example when all low variables are initialised in the code of the program; we will discuss this assumption in the next section. Thus, for such programs
\begin{quote}
{\bf Leakage:} (Shannon-based) leakage of a program $P$ is defined as the (Shannon) entropy of the partition $\Loi(P)$.
\end{quote}

We can now apply the results from section \ref{LAB:MAIN} in the context of programs, hence we deduce a correspondence between the refinement order of the observations, leakage, expected probability of guessing and expected number of guesses.

In terms of programs the results  from section \ref{LAB:MAIN} state the following equivalences:
\begin{enumerate}
\item $\Loi(P) \por \Loi(P')$
\item $\forall\mu. \ \ H_{\mu}(\Loi(P)) \leq H_{\mu}(\Loi(P'))$
\item  $\forall n, \mu. \ \ G_{n,\mu}(\Loi(P)) \leq G_{n,\mu}(\Loi(P'))$
\item   $\forall  \mu. \ \ NG_{\mu}(\Loi(P')) \leq NG_{\mu}(\Loi(P))$
\end{enumerate}
In words:
The equivalence relation associated to a program $P$ is refined by the equivalence relation associated to a program $P'$ if and only if for all distributions the leakage of $P$ is less than the leakage of $P'$, if and only if for any number of tries and any distribution the expected probability of guessing the secret is less according to $P$ than it is according to $P'$, if and only if for all distributions the expected number of guesses required to guess the secret according to $P$ is greater than the expected  number of guesses required to guess the secret according to $P'$.
\subsection{Relation with Yasuoka and Terauchi ordering results}
These order results are related to some recent work by Yasuoka and Terauchi \cite{ter1}; they define quantitative analysis in terms of Shannon entropy,  Smith's vulnerability and guessability.

 Their definitions follows the pattern we discussed before:
\begin{quote}
The quantitative analysis of confidential information of a program is defined as
the difference between an attacker's {\em capability}  before and after available observations about the program.
\end{quote}

By replacing the word ``capability" with: {\em (A) uncertainty about the secret, (B) probability of guessing the secret in one try, (C) expected number of guesses} we derive different quantitative analysis. Once formalized (A)(B)(C) as a function $F$  (and also its conditional counterpart $F(-|-)$ ) on a probability space all these definitions  will have the form:
\[ F(h|l) - F(h|l, \Loi(P)) \]
Formally the choices for $F,F(-|-)$ are:
\begin{enumerate}
\item[(A)] for {\em uncertainty about the secret} $F$ and $F(-|-)$ are Shannon entropy and conditional entropy
\item[(B)] for {\em probability of guessing in one try} (noted ME)
\[F(X)=-\log(\max _{x\in X}\mu(X=x))\mbox{ and }F(X|Y)=-\log(\sum_{y\in Y} \mu(y) (\max_{x\in X}\mu(X=x|Y=y  )) ) \]
\item[(C)] for {\em the expected number of guesses} (noted GE)
\[F(X)= \sum_{x_i\in X, i\ge 1}i \ \mu(X=x_i) \mbox{ and }F(X|Y)=\sum_{y\in Y} \mu(y) (\sum_{x_i\in X, i\ge 1} i \mu(X=x_i|Y=y  ))  \]
(assuming $i<j$ implies $\mu(X=x_i)\ge \mu(X=x_j)$)
\end{enumerate}

Shannon's entropy is unique in that
\begin{enumerate}
\item conditional mutual information is symmetric, so for $F$ being Shannon's entropy.
\[  F(h|l) - F(h|l, \Loi(P))= F(\Loi(P)|l) - F(\Loi(P)|l, h)  \]
and 
\item entropy of the result of a function given its arguments is 0 so 
\[ F(\Loi(P)|l,h ) = 0 \]
\end{enumerate}

In particular and again considering low inputs intialized in the program 
 it is only when $F$ is Shannon's entropy that 
\[ F(h) - F(h| \Loi(P))=F( \Loi(P)) \]
what this mean is that
\begin{quote}
It is only when using Shannon entropy that leakage as the difference in capability before and after observations is a measure on \latset
\end{quote}

We now want to relate results from  section \ref{LAB:MAIN} with ME and GE definitions of leakage.

To appreciate the difference in the definitions let's consider the examples from \cite{ter1}; we consider the following programs:
\begin{enumerate}
\item $M_1 \ \equiv \  {\tt  if (h== 1) o=0; else \ o=1;  } $
\item $M_2 \ \equiv \ {\tt  o=h;  }$
\end{enumerate}

Table \ref{Tera} shows the results of analyses of these programs for a 2 bits secret uniformly distributed. Columns H, G, NG corresponds to our definitions for Shannon entropy, the expected probability of guessing (in 1 guess) and the expected number of guesses on $\Loi(P)$, i.e.  
H, G, NG stands for $H(\Loi(P)), G(\Loi(P)), NG(\Loi(P))$.
ME and GE corresponds to the definitions in \cite{ter1} for computing the  min entropy and the guessing entropy on $P$; the final two columns ME' and GE' corresponds to apply the definitions in \cite{ter1} directly to $\Loi(P)$. For example $ME(M_1)=ME(h)-ME(h|\Loi(M_1))$ and $ME'(M_1)=ME(\Loi(M_1))$.

\begin{table}[htdp]
\caption{comparing measures}
\begin{center}
\begin{tabular}{|c|c|c|c|c|c|c|c|}
\hline
&H & G & NG & ME & GE & ME' & GE'  \\
\hline
$M_1 $& 0.8112 &  0.5 & 1.75 & 1 & 0.75 & 0.415 & 1.25\\
$M_2 $& 2 &  1 & 1 & 2 & 1.5  & 2 & 2.5\\
\hline
\end{tabular}
\end{center}
\label{Tera}
\end{table}%

The results express different ideas which can be connected in a uniform narrative. Take program $M_1:$ $G=0.5$  means after running the program an attacker has probability 0.5 of guessing the secret in one try. The chances of guessing the secret have doubled from 0.25 (before the program) to 0.5 (after the program), so the rate of increase is $2^{ME(M_1)}=2^1$; the average number of questions needed (initially 2.5) has been reduced by 0.75 (GE=0.75) so that it will take now on average to guess it NG=1.75 tries. And the observations provide 0.8112 bits of information about the secret.  

Consider now the second row, i.e. program $M_2$: here $H=2$ means that everything is leaked, i.e. the observations provide 2 bits of information about the secret. In this case we are sure to guess the secret in one try (G=1, NG=1) and our chances have hence increased 4 folds from the initial probabilities ($2^{ME(M_2)}=2^2$ so $0.25*2^{ME(M_2)}=1$);  the average number of questions needed (initially 2.5) has been reduced by 1.5 (GE=1.5) to one (NG=1). 

We have left out from the narrative the measures $ME',GE'$. The reason is that they seem of limited interest; for example $ME'$ will always pick the most likely observation and disregard all the others: a dubious security measure.

The narrative can be strengthened formally:
\begin{proposition}\label{ME}
For a program $P$
\begin{enumerate}
\item $\forall\mu.\  2^{(ME(P))} G(h)= G(\Loi(P))  $
\item $\forall\mu. \ GE(P)=NG(h)-NG(\Loi(P))  $

\end{enumerate}

\end{proposition}
\begin{proof}
(1) We start by recalling Smith's definition of vulnerability:
\[ ME(P)= \log\frac{1}{\max_h \mu(h)} -  \log\frac{1}{\sum_{o\in\Loi(P)} \max_h (h|o)} \]
We have then $\forall\mu $
\begin{eqnarray*}
 2^{ME(P)} G(h)&=& 2^{\log\frac{1}{\max_h \mu(h)} -  \log\frac{1}{\sum_{o\in\Loi(P)} \max_h (h|o)}}G(h)\\
&=& 2^{   \log{\sum_{o\in\Loi(P) } \max_h (h|o) }  - \log{\max_h \mu(h)} }G(h)\\
&=& \frac{2^{   \log{\sum_{o\in\Loi(P) } \max_h (h|o) } }} {2^{ \log{\max_h \mu(h)} }}G(h)\\
&=& \frac{\sum_{o\in\Loi(P) } \max_h (h|o) } {\max_h \mu(h) }G(h)\\
&=&	\frac{\sum_{o\in\Loi(P) } \max_h (h|o) } {\max_h \mu(h) }  \max_h \mu(h)\\
&=&	{\sum_{o\in\Loi(P) } \max_h (h|o) }\\
&=& G(\Loi(P))
\end{eqnarray*}
(2)
We can rewrite the definition of $GE(P)$ from \cite{ter1} as:
\[   GE(P)=\sum_{1\leq i\leq n}  i\mu(h_i) - \sum_{o\in \Loi(P)} \sum_{h_i\in o, 1\leq i\leq m} i\mu(h_i)  \]
It is easy to see that the first term coincides with our definition on $NG$ on sets and the second term with our definition of $NG$ on partitions; the result then follows.
\end{proof}

The connections between these concepts extends to the orders they induce:

\begin{theorem}
Given programs $P,P'$ (non depending on the low inputs) the following are equivalent:
\begin{enumerate}
\item $\Loi(P) \por \Loi(P')$
\item $\forall \mu.\ \Loi(P) \leq_d \Loi(P')$
\item $\forall\mu. \ \ H_{\mu}(\Loi(P)) \leq H_{\mu}(\Loi(P'))$
\item $\forall n, \mu. \ \ G_{n,\mu}(\Loi(P)) \leq G_{n,\mu}(\Loi(P'))$
\item $\forall  \mu. \ \ NG_{\mu}(\Loi(P')) \leq NG_{\mu}(\Loi(P))$
\item $\forall \mu. \ \ ME_{\mu}(P) \leq ME_{\mu}(P')$
\item $\forall \mu. \ \ GE_{\mu}(P) \leq GE_{\mu}(P')$
\end{enumerate}
\end{theorem} 

\begin{proof}
equivalence $1\Leftrightarrow 3$ was first proved in \cite{qa09}, equivalences $1\Leftrightarrow 2,4,5$ proved in section \ref{LAB:MAIN};
equivalences   $1\Leftrightarrow 3,6,7$ are proven in \cite{ter1}.
It may be however interesting to reprove the equivalences  in \cite{ter1} using the algebraic techniques and results from this paper. For example we can prove $1\Leftrightarrow 6$ as follows:
\begin{eqnarray*}
\Loi(P) \por \Loi(P') & \Leftrightarrow& \forall\mu.\ G(\Loi(P)) \leq G(\Loi(P'))\\
& \Leftrightarrow& \forall\mu.\  2^{(ME(P))} G(h)  \leq  2^{(ME(P'))} G(h)  \\
& \Leftrightarrow&\forall\mu.\  ME(P) \leq ME(P')
\end{eqnarray*}
where the first equivalence is corollary \ref{Lab:LoIExpProb2} and the second is proposition \ref{ME}.

$1\Leftrightarrow 7$ follows from Proposition \ref{ME}(2):just rewrite it as $GE(P)=NG(h)-NG(\Loi(P))$
\end{proof}

Hence we conclude that  in terms of the induced orderings all these quantitative analyses are consistent. In other words it is only on programs not ordered on $\latset$ that these notions can really differ.
One such difference is now discussed.

\subsection{Discussion on Smith's argument on the foundations of Quantitative Information Flow}
Consider  the following two programs \cite{Fossacs09}:
\begin{enumerate}
\item $P_1 \ \equiv \  {\tt if\ (h \% 8==0)\  o= h; \ else\ o=1; } $
\item $P_2 \ \equiv \ {\tt o \ =\ h\& 037; }$
\end{enumerate}

The program $P_1$ will return the value of $h$ when the last three bits of the secret are 0s and will return 1 otherwise; 
its \latset\ interpretation will hence be  the partition of the form

\[ X=\{ \{h_1000\},\dots,\{h_m000 \} , X_1 \}  \]
where the $h_i$ are arbitrary binary string of length $k-3$.

The program $P_2$ copies the last 5 bits of the secret in $o$ (here $037$ is the octal constant and $\&$ the bitwise and).
The partition associated has hence the shape
\[ Y= \{    Y_1,\dots, Y_r  \}  \] 
where each $Y_i$ is a set of string with the same 5 last bits.

Smith's argument is that under uniform distribution and for a secret of size $8k$ bits the two programs have a very similar entropy but they have a very different guessing behaviour; in the case of the first program in fact with probability one eight  the whole secret is revealed, while in the second program all attempts reveal the last 5 bits of the secret but give no indication of what the remaining bits are. Hence in general it is much easier to guess the secret in one try after running  the first program than it is to guess the secret after running the second one.

The argument however relies on choosing a particular distribution; this choice is independent from the source code and should, we believe, be clearly separated from the leakage inherent to the code. 

In fact since
the partitions $X,Y$ are unrelated in \latset, by the  results from section \ref{LAB:MAIN}  we can find distributions and number of guesses that make one's expected guessing probability less than the other.

For $G_n(X)<G_n(Y)$ notice that $X_1$ splits in many blocks $Y_i$: hence take any distribution non zero only on the atoms in $X_1$ e.g. let's consider the uniform distribution on the atoms in $X_1$ and take $n=|X_1|-1$.

Then 
\[G_n(X)=g_n(X_1)=\frac{n}{n+1}<1=G_n(Y)  \]

To make $G_n(X)>G_n(Y)$ pick any block $Y_i$ in $Y$ whose last three bits are 0s; then this block is split in many $X_i$s in $X$, again by taking the distribution uniform over the elements of $Y_i$ and 0 otherwise and taking $n=|Y_i|-1$ we have
\[G_n(Y)=g_n(Y_1)=\frac{n}{n+1}<1=G_n(X)  \]
In fact  all distributions giving probability 0 to all values divisible by 8 will favour program $P_2$ even when we consider a single guess (n=1), and things don't change when we take $ME$ instead of $G_n$.

Similarly we can find distributions that make the expected number of guesses of any of the two programs less than the expected number of guesses of the other program. In particular while for the uniform distribution it is much easier to guess the secret in the case of the first program compared to the second (which is at the heart of Smith's argument), by choosing the distribution zero everywhere apart from the block $X_1$ it become easier to guess the secret using the second program.
While such a distribution may be seen as pathological it still shows the possible problems in making code analysis dependent on particular distributions.

 \subsection{\latset, maximum leakage and Channel Capacity}
The relation between \latset\ and channel capacity has been investigated in the literature 
\cite{PLAS08,ter2, KoSm}. The channel capacity of a program is defined as the maximum possible leakage for that program. Intuitively this is the context most advantageous for the attacker. \latset\ provides an elementary characterization of channel capacity: in fact  as the leakage is defined by $H(\Loi(P))$ using the well known information theoretical fact that the maximal entropy over a system with $n$ probabilities is $\log(n)$ we deduce that the channel capacity is $\log(|\Loi(P)|)$.
 
We note by ${\tt CC}(P)$ for the channel capacity of the program $P$. We have then
 
  \begin{proposition}\label{LAB:CC}
  \[ \Loi(P) \sqsubseteq \Loi(P') \Rightarrow {\tt CC}(P)\leq {\tt CC}(P') \]
   \end{proposition}

If  $\Loi(P)\sqsubseteq \Loi(P') $ then all blocks of $\Loi(P)$ are refined by blocks of $\Loi(P')$ so the number of blocks of $\Loi(P)$ is $\leq $ than the number of blocks of $\Loi(P')$, but the channel capacity for programs is the log of the number of blocks interpretation, hence the result  is proved.

The opposite direction of the implication doesn't hold: for example the partitions
\[  \{ \{a,b,c\},\{d\}  \} \mbox{ and } \{ \{a,b\},\{c,d\} \}\]
are not order related but have the same channel capacity 1.

 \subsubsection{\latset\ and min-entropy Channel Capacity}
 The relation between channel capacity of a program $P$ and $\log(|\Loi(P)|)$ is not confined to Shannon entropy. In fact K\"opf and Smith have shown that even if we choose Smith's min-entropy quantitative analysis \cite{KoSm} we get the same value, i.e. maximum vulnerability of a program $P$   according to Smith measure $ME$ is $\log(|\Loi(P)|)$. We hence have the equalities
 \[  {\tt CC}(P) = \log(|\Loi(P)|) =\max_{\mu} H_{\mu}(\Loi(P))=\max_{\mu} ME_{\mu}(P)\]

\section{Low inputs, Multiple runs and l.u.b. in  \latset}\label{LOW}
A major source of confusion in security analysis derives from poorly defined attacker models.
In this section we discuss a few  common modelling issues and how they can be dealt with in \latset.
\subsection{Active and passive attackers}
The lattice of information allows for different attacker's models: the most common and possibly interesting is the one corresponding to an {\em active attacker}, i.e. an attacker who control the low inputs; a typical example would be a cash machine where an attacker is able to choose a pin number. 
An active attacker can be modelled as we did in the previous sections by assuming that the low variables are initialised in the code, the initialisation values corresponding to the attacker choice.

We could however also model a {\em passive attacker}, an eavesdropper with no power to choose the low inputs. In this case the lattice atoms are the pair of low and high inputs. Take for example the program
 
 \begin{verbatim}
if (h == l)  o= 1; else o=2;
\end{verbatim}

where $h,l$ are 2 bits variables. The partition associated to the programs is:
\[ 
\{ \{ (0,0) \}, \{  (0,1),(0,2),(0,3) \}  , \dots 
,\{ (3,3)  \}, \{(3,0),(3,1),(3,2) \}  \} 
\]
assuming uniform distribution on the low and high inputs we then compute leakage as
\[H(\Loi(P)|l)= 4 \frac{1}{8} \log(4)+4\frac{1}{8}\log(\frac{4}{3}) =0.60375 \]

In fact an active attacker is a particular case of this setting, where the distribution on the inputs is such that only one low input has probability non-zero.
In that case the atoms of the lattice are, up to isomorphism, only the high inputs and $H(\Loi(P)|l)=H(\Loi(P))$.

\subsection{Non termination} In this work we have mostly considered output observations as values. We can however relax this and include among the possible observations non termination. 
This doesn't change the theory: non-termination is just an additional equivalence class: the class of all input states over which the program doesn't terminate; of course the usual computational and complexity problems arise when we try to compute such a class.

\subsection{Multiple runs}
Another aspect of an attacker model that has a natural algebraic interpretation in $\latset$ is an attacker capability to run the system $n$ times: for example an attacker trying three pin numbers on a cash machine.
Running a program several times with different low inputs may reveal more and more information about the secret; For example consider the password checking program $P$

\begin{verbatim}
if (h == l)  o= 1; else o=2;
\end{verbatim}

If we run it once assigning the value 5 to the low variable we gain the information whether the secret is 5 or not; by running it twice, assigning  to the low variable the value 5 and the value 7 we will gain the information whether the secret is 5 or is 7 or something else.

Written in terms of partitions this is nothing else than the join operation in \latset
\[  \{\{ 5 \} ,\{\not= 5  \}\} \sqcup \{\{  7\}, \{ \not=7 \}\}   =\{\{ 5 \} ,\{  7 \} , \{  \not=5,7 \}\}        \]

Hence the knowledge available to an attacker who can choose the low inputs and run the program $m$ times is modelled by the partition
\[  \Loi(P_1) \sqcup \dots \sqcup  \Loi(P_m)\]
where $ \Loi(P_i)$ is the partition corresponding to the i-th run of the program.

\subsubsection{Does it leak the same information?} A related question is whether a program leaks always the same information for each run of the program; for example a program leaking the last bit of the secret always leaks the same information no matter how many times we run the program but a password check leaks different information when we run it choosing different low inputs. This question can also be addressed by using l.u.b.: if the program $P$ leaks different information over different runs this means we can find two runs $P_i,P_j$ such that
\[   \Loi(P_i) \sqcup  \Loi(P_j) >    \Loi(P_i),  \Loi(P_j)\]
The interpretation of multiple runs in terms of l.u.b.s has also somehow a reverse implication, i.e. it is possible, given programs  $P_1,P_2$ to build a program whose interpretation is their l.u.b. This result has a practical significance: when  $P_1,P_2$ are different runs of the same program the  l.u.b. is their self-composition \cite{rezk}.  Formally \cite{bert}: 

\begin{proposition}\label{LUB}
Given programs $P_1,P_2$ there exists a program $P_{1\join 2}$ such that \[{\Loi}(P_{1\join 2})={\Loi}(P_1)\join {\Loi}(P_2)\]
\end{proposition}
Given programs $P_1, P_2$, we define $P_{1\join 2} = P'_1; P'_2$ where the primed programs $P'_1, P'_2$ are $P_1, P_2$ with variables renamed so to have disjoint variable sets. If the two programs are syntactically equivalent, then this results in self-composition \cite{rezk} For example, consider the two programs
%\begin{equation*}
\[ P_1\equiv {\tt if\ (h== 0)\  x= 0;\ else\ x=1;}, \ \ \ P_2\equiv {\tt  if\ (h== 1)\ x= 0; \ else\ x=1;}  \]
%\end{equation*}
with their partitions $\Loi(P_1) =  \{\{ 0\}, \{ h\not=0 \} \}$ and $\Loi(P_2) = \{ \{ 1\} ,\{ h\not=1 \} \}$. The program  $P_{1\join 2}$ is the concatenation of the previous programs with variable renaming
\begin{eqnarray*}
P_{1\join 2} &\equiv& {\tt h' = h; if\ (h'== 0)\  x'= 0;\ else\ x'=1;}\\
&\ & {\tt h'' = h; if\ (h'' == 1)\ x'' = 0; \ else\ x'' =1;}
\end{eqnarray*}
The corresponding lattice element is the join, i.e. intersection of blocks, of the individual programs $P_{1} ,P_{2}$
%\begin{equation*}
\[ \Loi(P_{1\join 2}) = \{\{ 0\}, \{1\}, \{ h \not=0,1 \} = \{\{ 0\}, \{ h\not=0 \} \} \join \{ \{ 1\}, \{ h\not=1 \} \} \]
%\end{equation*}

\section{Further applications of \latset}\label{LOOP}
We quickly review two applications of \latset\ beyond the foundational aspect:
\subsection{Loop analysis} Loop constructs are difficult to analyse. However they have a natural interpretation in the lattice of information. In informal terms the idea is that loops can be seen as l.u.b. of a chain  in the lattice of information, where the chain is the interpretation of the different iterations of the loop.

To understand the ideas let's consider the program
\begin{quote}
\begin{verbatim}
l=0; 
while(l < h) {
   if (h==2) l=3 else l++ 
} 
\end{verbatim}
\end{quote}
and let us now study the partitions it generates.
The loop terminating in 0 iterations will reveal that {\tt h=0} i.e. the partition $W_0=\{\{0\}\{1,2,3\}\}$, termination in 1 iteration will reveal {\tt h=1} if  the output is 1 and  {\tt h=2} if the output is 3 i.e. 
 $W_1=\{\{1\}\{2\}\{0,3\}\}$,  the loop will never terminate in 2 iterations i.e.  $W_2=\{\{0,1,2, 3\}\}$ and in 3 iterations will reveal that  {\tt h=3} given the output 3, i.e.  $W_3=\{\{3\}\{0,1,2\}\}$.
  Let's define $W_{\leq n}$ as $\join_{n\ge i\ge 0} W_{i}$; we have then the chain\footnote{the chain is trivial in this example}
  \[ W_{\leq 1} =W_{\leq 2} =W_{\leq 3}=\{\{0\}\{1\}\{2\}\{3\} \} \]
We also introduce an additional partition $C$ to cater for the collisions in the loop: the collision partition is $ C=\{ \{0\}\{1\}\{2,3\}  \}$ because for {\tt h=2} the loop terminates with output 3 in 1 iterations and for {\tt h=3} the loop terminates with output 3 in 3 iterations. 
We have then 
\[ \Loi(P)=\join_{n\ge 0} W_{\leq n} \meet C   = \{ \{0\}\{1\}\{2,3\}  \}  \]

This setting is formalized in \cite{bert}. 
Given a looping program $P$ define $W_{\leq n}$ as the equivalence relation corresponding to the output observations available for the loop terminating in $\leq n$ iterations and let the  {\em collision equivalence} of a loop be the reflexive and transitive closure of the relation  $\sigma\simeq_C\sigma'$ iff  $\sigma ,\sigma'$ generate the same output from different iterations.

The following is then true:
\begin{proposition}
\[ \Loi(P)=\join_{n\ge 0} W_{\leq n} \meet C\]
\end{proposition}
Hence leakage $H(\Loi(P))$ for looping programs can be computed in terms of the chain $(W_{\leq n})_{n\ge 0}$ and the collision equivalence $C$.
The equivalence of this technique with previous information theoretical analysis of loops \cite{JCS09} is proved in \cite{bert}.
\subsection{Analysis of C-code vulnerabilities}
Recent work  \cite{acsac}  based on the $\Loi$ interpretation of programs, has demonstrated the applicability of QIF to real world vulnerabilities. Previous attempts to implement a quantitative analysis had hit a major hurdle: in very simple terms since QIF is based on $\Loi(P)$ and $\Loi(P)$ is the set theoretical kernel of the denotational semantics of $P$  computing  $\Loi(P)$ is computationally unfeasible.  
The approach followed in  \cite{acsac} is to change the QIF question from computing $\Loi(P)$ to computing bounds on the channel capacity ${\tt CC}(P)$. We saw these concepts are related in theorem \ref{LAB:CC}. 
Using assume-guarantee reasoning questions about bounds can be expressed in verification terms\footnote{A similar idea has been independently proposed by Yasuoka and Terauchi in \cite{ter2}}. In particular by expressing them as drivers for  the symbolic model checker  CBMC \cite{ckl2004} several CVE reported vulnerabilities in the Linux kernel were quantitatively analysed in \cite{acsac}; moreover the official patches for such vulnerabilities were formally verified as fixing the leak.
That work is the first demonstration of quantitative information flow addressing security concerns of real-world industrial programs.

\section{Conclusions}
We investigated the importance of the Lattice of Information for Quantitative Information Flow. This lattice allows for an algebraic treatment of confidentiality and clarifies the relationship between the Information Theoretical, probabilistic and guessability measures that are used in QIF.
Our results show that  these measures  are all consistent w.r.t. the classification of language based confidentiality threats,
and this classification is captured by the refinement order in \latset.

We have seen how these results fit and contribute to recent work in the community, especially the ones by Yasuoka and Terauchi \cite{ter1} and by Smith \cite{Fossacs09}. 
It is a matter for future research to determine whether the Lattice of Information can also provide a unifying foundation for non-deterministic and probabilistic systems.

\subsection{Acknowledgements}
I am very grateful to Jonathan Heusser with whom work reported in section \ref{BASIC} and section \ref{LOOP} was carried on.


\begin{thebibliography}{}
\bibitem[BKR]{koepf09}
Michael Backes and Boris K{\"o}pf and Andrey Rybalchenko:
Automatic Discovery and Quantification of Information Leaks.
Proc. 30th IEEE Symposium on Security and Privacy (S\& P '09), to appear
\bibitem[GAR]{rezk}
Barthe, Gilles and D'Argenio, Pedro R. and Rezk, Tamara:
Secure Information Flow by Self-Composition.
CSFW '04: Proceedings of the 17th IEEE workshop on Computer Security Foundations.
\bibitem[B]{birkhoff}
Birkhoff, G., Lattice theory. Amer. Math. Soc. Colloq. Publ. 25 (1948).
\bibitem[CPP]{cpp1}
K. Chatzikokolakis, C. Palamidessi, P. Panangaden.  Anonymity protocols as noisy channels. Information and Computation, 206(2-4):378-401, 2008.
\bibitem[CM1]{chen}
Han Chen, Pasquale Malacaria:
Quantitative analysis of leakage for multi-threaded programs.
PLAS '07: Proceedings of the 2007 workshop on Programming languages and analysis for security.
\bibitem[CM2]{ASIACCS09}
Han Chen, Pasquale Malacaria:
Quantifying Maximal Loss of Anonymity in Protocols.
In Proceedings ACM Symposium on Information, Computer and Communication Security 2009.
\bibitem[CHM1]{qapl}
David Clark, Sebastian Hunt, Pasquale Malacaria:
Quantitative Analysis of the Leakage of Confidential Data.
Electronic Notes in Theoretical Computer Science, Volume 59, Issue 3, QAPL'01, Quantitative Aspects of Programming Laguages, November 2002, Pages 238-251.
\bibitem[CHM2]{0dav}
David Clark, Sebastian Hunt, Pasquale Malacaria:
A static analysis for quantifying information flow in a simple imperative language.
Journal of Computer Security, Volume 15, Number 3 / 2007.
\bibitem[CHM3]{CHM05}
David Clark, Sebastian Hunt, and Pasquale Malacaria:
Quantitative information flow, relations and polymorphic types. Journal of Logic and Computation, Special Issue on Lambda-calculus, type theory and natural language, 18(2):181-199, 2005.
\bibitem[CKL]{ckl2004}
Clarke, Edmund, and Kroening, Daniel, and Lerda, Flavio:
A Tool for Checking {ANSI-C} Programs.
Tools and Algorithms for the Construction and Analysis of Systems (TACAS 2004).
Springer, 168--176, Volume 2988
\bibitem[CMS]{belief}
{Clarkson, Michael R. and Myers, Andrew C. and Schneider, Fred B.:}
{Quantifying information flow with beliefs.}
{J. Comput. Secur., 17(5):655-701, 2009.}
\bibitem[CT]{CoTho}
Thomas M. Cover and Joy A. Thomas:
Elements of Information Theory.
John Wiley, 1991
\bibitem[KB]{koepfbasin}
Boris K\"{o}pf and David Basin:
An information-theoretic model for adaptive side-channel attacks.
CCS '07: Proceedings of the 14th ACM conference on Computer and communications security, 2007, 286-296
%\bibitem[aaa]{CoTho} T.Cover, J. Thomas. Elements of Information Theory. Wiley
\bibitem[KS]{KoSm}
Boris K\"{o}pf, Geoffrey Smith: Vulnerability Bounds and Leakage Resilience of Blinded Cryptography under Timing Attacks. CSF 2010: 44-56
\bibitem[D]{1den}
Denning, Dorothy E.
A lattice model of secure information flow.
Commun. ACM, 19, 5, 1976, 236--243
ACM, New York, NY, USA
\bibitem[G]{gray}
J.W.Gray III:
Toward a mathematical foundation for information flow security.
Proc. 1991 IEEE Symposium on Security and Privacy, Oakland, CA, 1991, pp. 21Ð34.
\bibitem[HM1]{acsac}
Jonathan Heusser, Pasquale Malacaria: Quantifying Information Leaks In Software. In Proceedings of the ACM Annual Computer Security Applications Conference, ACSAC 2010, Austin, Texas, USA, December 6-10 2010. ACM 2010
\bibitem[HM2]{aqua}
Jonathan Heusser and Pasquale Malacaria.
Applied Quantitative Information Flow and Statistical Databases.
In Proceedings of Workshop on Formal Aspects in Security and Trust (FAST 2009), 2009.
\bibitem[HM3]{qa09}
Jonathan  Heusser  and Pasquale Malacaria. Quantifying Loop Leakage using a Lattice of Partitions. In Proceedings of the 1st Workshop on Quantitative Analysis of Software (QA'09), 2009.
\bibitem[LR]{loi}
Landauer, J., and Redmond, T.:
A Lattice of Information.
In Proc. of the IEEE Computer Security Foundations Workshop.
IEEE Computer Society Press, 1993.
\bibitem[M1]{POPL07}
Pasquale Malacaria:
Assessing security threats of looping constructs.
Proc. ACM Symposium on Principles of Programming Language, 2007.
\bibitem[M2]{JCS09}
Pasquale Malacaria:
Risk Assessment of Security Threats for Looping Constructs,  Journal Of Computer Security, 18(2), 2010 .
\bibitem[MC]{PLAS08}
Pasquale Malacaria, Han Chen:
Lagrange Multipliers and Maximum Information Leakage in Different Observational Models.
ACM SIGPLAN Third Workshop on Programming Languages and Analysis for Security. June, 2008.
\bibitem[MH]{bert}
Pasquale Malacaria, Jonathan Heusser: Information Theory and Security: Quantitative Information Flow. In 10th International School on Formal Methods for the Design of Computer, Communication and Software Systems, SFM 2010, Bertinoro, Italy, June 21-26, 2010, Advanced Lectures. Lecture Notes in Computer Science 6154 Springer 2010: 87-134
\bibitem[M]{Massey}
James L. Massey:
Guessing and entropy.
In Proceedings of the 1994 IEEE International Symposium on Information Theory, 1994, 204
\bibitem[MM]{mciver}
A.McIver and, C.Morgan:
A probabilistic approach to information hiding
Programming Methodology, Springer, New York, NY, USA, 2003, pp. 441Ð460.
\bibitem[Mc]{1mcl}
John Mclean:
Security Models and Information Flow.
In Proc. IEEE Symposium on Security and Privacy, 1990, 180--187
IEEE Computer Society Press
\bibitem[M]{1mil}
Jonathan K. Millen:
Covert Channel Capacity.
IEEE Symposium on Security and Privacy, 0, 1987, 1540-7993, 60
IEEE Computer Society, Los Alamitos, CA, USA
\bibitem[MC]{mu}
C. Mu and D. Clark:
An Abstraction Quantifying Information Flow over Probabilistic Semantics.
Workshop on Quantitative Aspects of Programming Languages (QAPL), ETAPS, 2009.
\bibitem[N]{SemVal}
Y. Nakamura.
Entropy and Semivaluations on Semilattices. Kodai Math. Sem. Rep
22 (1970), 443 468
\bibitem[R]{Renyi}
A. R\'enyi:
On measures of information and entropy.
Proceedings of the 4th Berkeley Symposium on Mathematics, Statistics and Probability 1960: 547-561.
\bibitem[S1]{shannon}
{Claude E. Shannon:}
{A mathematical theory of communication.}
{Bell Systems Technical Journal, 27(3):379-423, 1948}
\bibitem[S2]{Sh2}
{Claude E. Shannon:}
{The lattice theory of information.}
{IEEE Transactions on Information Theory, 1:105-107, 1953}
\bibitem[Sm]{Fossacs09}
Geoffrey Smith:
On the Foundations of Quantitative Information Flow.
In Proc. FOSSACS 2009: Twelfth International Conference on Foundations of Software Science and Computation Structures
LNCS 5504, pp. 288-302, York, UK, March 2009
\bibitem[TA]{safety}
T. Terauchi and A. Aiken.
Secure information flow as a safety problem:
In SAS, volume 3672 of LNCS, pages 352--367, 2005.
\bibitem[YT1]{ter1}
Hirotoshi Yasuoka, Tachio Terauchi: Quantitative Information Flow - Verification Hardness and Possibilities. In Proceedings CSF 2010: 15-27
\bibitem[YT2]{ter2}
Hirotoshi Yasuoka and Tachio Terauchi.
On Bounding Problems of Quantitative Information Flow. In Proceedings
ESORICS 2010.
\bibitem[W]{W93}
Glynn Winskel.
The Formal Semantics of Programming Languages. The MIT Press 1993.

%\bibitem[aaa]{weber}
%D.G. Weber:
%Quantitative Hook-Up Security for Covert Channel Analysis
%In Proc. IEEE Computer Security Foundations Workshop 1988.
\end{thebibliography}
\end{document}